\documentclass[12pt]{amsart}
\usepackage{amsmath,amsthm,amssymb,amsfonts,verbatim}
\usepackage[margin =1.1in]{geometry}

\usepackage{hyperref}

 \providecommand{\F}{\mathbb{F}}

\title[Optimal locally recoverable codes of distance $5$ and $6$]{Explicit construction of optimal locally recoverable codes of distance $5$ and $6$ via binary constant weight codes }

\author{ Lingfei Jin}

\thanks{Lingfei Jin is with Shanghai Key Laboratory of Intelligent Information Processing, School of Computer Science, Fudan University, Shanghai 200433, China. {\it email:} {lfjin@fudan.edu.cn}. }

\newtheorem{lemma}{Lemma}[section]
\newtheorem{theorem}[lemma]{Theorem}
\newtheorem{cor}[lemma]{Corollary}

\newtheorem{ex}[lemma]{Example}
\newtheorem{defn}{Definition}
\newtheorem{rmk}{Remark}

\renewcommand{\epsilon}{\varepsilon}
\renewcommand{\le}{\leqslant}
\renewcommand{\ge}{\geqslant}
\def\F {{\mathbb{F}}}

\def \Xi {{X^{[i]}}}

\newcommand{\Ga}{\alpha}
\newcommand{\Gb}{\beta}
     
\newcommand{\Gd}{\delta}

\newcommand{\Gl}{\lambda}

\def \a {{\bf a}}
\def\b{{\bf b}}

\def \bc {{\bf c}}
\def \bi {{\bf 1}}

\def \bh {{\bf h}}

\def \bu {{\bf u}}
\def \bv {{\bf v}}
\def \bo {{\bf 0}}

\begin{document}

\maketitle

\begin{abstract} It was shown in \cite{GXY18} that the length $n$ of a $q$-ary linear locally recoverable code with distance $d\ge 5$ is upper bounded by $O(dq^3)$. Thus, it is a challenging problem to construct  $q$-ary locally recoverable codes with distance $d\ge 5$ and length approaching the upper bound.  The paper \cite{GXY18} also gave an algorithmic construction of $q$-ary locally recoverable codes with locality $r$ and length $n=\Omega_r(q^2)$ for $d=5$ and $6$, where $\Omega_r$ means that the implicit constant depends on locality $r$. In the present paper, we present an explicit construction of $q$-ary locally recoverable codes of distance $d= 5$ and $6$ via binary constant weight codes. It turns out that (i) our construction is simpler and more explicit; and (ii) lengths of our codes are larger than those given in \cite{GXY18}.
\end{abstract}

\section{Introduction}
Motivated by the problem of designing efficient codes for distributed storage systems,  locally recoverable (or repairable) codes  have recently attracted great attention of researchers. A local repairable code is  a  block code with an additional parameter called {\it locality}.

A block code is called a locally repairable code (LRC for short) with locality $r$ if every symbol in the encoding is a function of
$r$ other symbols (see the precise definition of locally repairable codes in Section \ref{subsec:2.1}). This enables recovery of any single erased symbol in a local fashion by downloading at most $r$ other symbols. On the other hand, one would like the code to have a good minimum distance to enable recovery of many erasures in the worst-case. LRCs have been the subject of extensive study in recent years~\cite{HL07,GHSY12,PKLK12,SRKV13,JMX17,LMX17,FY14,PD14,TB14,TPD16,BTV17,KBTY18}.
LRCs offer a good balance between very efficient erasure recovery in the typical case in distributed storage systems where a single node fails (or becomes temporarily unavailable due to maintenance or other causes), and still allowing recovery of the data from a larger number of erasures and thus safeguarding the data in more worst-case scenarios.

A Singleton-type bound for LRCs relating its
length $n$, dimension $k$, minimum distance $d$ and locality $r$ was first shown in the highly influential work~\cite{GHSY12}. It states that a linear locally repairable code $C$ must obey
\begin{equation}\label{eq:x1}
 d(C)\le n-k-\left\lceil \frac kr\right\rceil+2.
\end{equation}
%
Note that any linear code of dimension $k$ has locality at most $k$, so in the case when $r=k$ the above bound specializes to the classical Singleton bound $d \le n-k+1$, and in general it quantifies how much one must back off from this bound to accommodate locality.

A linear LRC that meets the bound \eqref{eq:x1} with equality is said to be an \emph{optimal} LRC. This paper focuses on the trade-off between alphabet size and code length for linear codes that are optimal LRCs.
One is thus tempted to make an analogy between optimal LRCs and MDS codes. The famous MDS conjecture says that there are no non-trivial (meaning, distance $d > 2$) MDS codes with length exceeding $q+1$ where $q$ is its alphabet size, except in two corner cases  ($q$ even and $k=3$, or $k=q-1$) where the length is at most $q+2$. This conjecture was famously resolved in the case when $q$ is prime by Ball \cite{Ball}.

In view of the result given in \cite{GXY18}, we define the following.
\begin{defn}{\rm Given a prime power $q$, locality $r\ge 2$ and $d\ge 5$, define
\[
N_q(d,r)=\max\{n\ge 2:\; \mbox{there exists a q-ary linear optimal LRC of length $n$, distance $d$ and locality $r$} \}.
\]
}\end{defn}
The main purpose of this paper is to give an explicit construction of LRCs that provides  lower bounds on $N_q(5,r)$ and $N_q(6,r)$.
\subsection{Known results}
The early constructions of optimal LRCs produced codes with alphabet size that is exponential in code length (see \cite{HCL,SRKV13}). In \cite{PKLK12}, another construction of optimal LRCs was proposed with alphabet  size comparable to code length. But the construction in \cite{PKLK12} only produced  a specific value of the length $n=\left\lceil \frac kr\right\rceil(r+1)$ which indicates  the rate of the code is very close to $1$. Besides, there are some existence results showed in \cite{PKLK12} and \cite{TB14} where there are less restriction on locality $r$. However, large alphabet size which is an exponential function of the code length  is required for those results.
A breakthrough construction  given  in \cite{TB14} produced  optimal LRCs with length linear  in  alphabet size although the length of codes is upper bounded by alphabet size. The idea is to use subcodes of Reed-Solomon codes. Few years later, the construction in \cite{TB14} was extended in \cite{JMX17} by using automorphism group of rational function fields. It turns out that more flexibility on locality can be achieved  and the code length can go up to $q+1$, where $q$ is the alphabet size.

Similar to the classical MDS conjecture,  one natural question is  whether the length of a $q$-ary optimal  LRC can go beyond $q+1$. Recently, this question was positively answered in \cite{LMX17}. By using elliptic curve, it was shown that there exist $q$-ary optimal  LRCs with length $n$ bigger than $q+1$ and distance linear in length $n$. More surprisingly, it was shown in \cite{LXY} that there exist $q$-ary optimal LRCs  of distance $3$ or $4$ and arbitrarily large length, i.e., there is a family of optimal LRCs  of distance $3$ or $4$ with length tending infinity.
Very recently, it was shown in \cite{GXY18} that the length of an optimal linear LRC of distance $d\ge 5$ is upper bounded by $O(dq^3)$. In particular, the length of an optimal linear LRC of distance  $5$ (and $6$, respectively) is upper bounded by $O(q^2)$ (and $O(q^{3})$, respectively). Thus, it is a challenging problem to construct  $q$-ary LRCs with distance $d\ge 5$ and length approaching the upper bound.

Furthermore, the paper \cite{GXY18}  gave an algorithmic construction of $q$-ary optimal LRCs with locality $r$ and length $n=\eta q^2$ subject to the constraints $r\ge d-1$, $(r+1)|(n+1)$ and $\eta\le \left(\frac1{2(d-1)^{d-1}(r+1)^{(d-1)/2}}\right)^{1/\lfloor\frac{d-3}2\rfloor}$. Thus, the paper \cite{GXY18} produced
(i) a $q$-ary LRCs with locality $r\ge 4$, distance $d=5$ and length $n=(r+1)\lfloor \frac{q^2}{2^{10}\times (r+1)^3}\rfloor\approx \frac{q^2}{2^{10}\times (r+1)^2}$; and (ii) a $q$-ary LRCs with locality $r\ge 5$, distance $d=6$ and length $n=(r+1)\lfloor \frac{q^2}{2\times 5^5\times (r+1)^{3.5}}\rfloor\approx \frac{q^2}{2\times 5^5\times (r+1)^{2.5}}$. Precisely speaking, \cite{GXY18} gave the following lower and upper bounds.

\begin{lemma}\label{lem:1.1} If $(r+1)|n$ and $r\ge 5$, then
\[(r+1)\left\lfloor \frac{q^2}{2^{10}\times (r+1)^3}\right\rfloor\approx \frac{q^2}{2^{10}\times (r+1)^2}\le N_q(5,r)\le\frac{r+1}r\cdot\frac{q}{q-1}\cdot q^2\]
and
\[(r+1)\left\lfloor \frac{q^2}{2\times 5^5\times (r+1)^{3.5}}\right\rfloor\approx \frac{q^2}{2\times 5^5\times (r+1)^{2.5}}\le N_q(6,r)\le\frac{r+1}r\cdot\frac{q}{q-1}\cdot q^3.\]
\end{lemma}

\subsection{Our result and comparison}
By making use of parity-check matrices, we present an explicit construction of optimal LRCs of distance $5$ and $6$. The key ingredients of our construction is a family $\{I_i\}_{i=1}^m$ of subsets of $\{1,2,\dots,n\}$ such that $|I_i|=r+1$ and $|I_i\cap I_j|\le 1$ for all $1\le i<j\le m$. This turns out that we require   a binary constant weight code with specific parameters. Precisely speaking, our main result is
\begin{theorem}\label{thm:1.2} One has the following results.
\begin{itemize}
\item[{\rm (i)}] {\rm (see Corollary \ref{cor:3.3})} If $r+1\ge 5$ is a prime power, then for any $t\ge 1$ there exists an explicit construction of a family of optimal $q$-ary $[n,k,5]$-LRCs  with locality $r$, where $q=(r+1)^t$, $n=(r+1)m$ and $k=n-m-3$ and $m=\frac{(r+1)^{t-1}((r+1)^t-1)}{r}$. Hence, $N_q(5,r)\ge \frac 1r q(q-1)$.
\item[{\rm (ii)}] {\rm (see Corollary \ref{cor:3.4})} If $r+1\ge 8$ is a power of $2$, then for any $t\ge 1$ there exists an explicit construction of a family of optimal $q$-ary $[n,k,6]$-LRCs  with locality $r$, where $q=(r+1)^t$, $n=(r+1)m$ and $k=n-m-4$ and $m=\frac{(r+1)^{t-1}((r+1)^t-1)}{r}$. Hence, $N_q(6,r)\ge\frac 1r q(q-1)$.
\item[{\rm (iii)}] {\rm (see Corollary \ref{cor:3.5})}
For $4\le r\le q-1$ with $(r+1)|n$,  there exists an explicit construction of an optimal $q$-ary $[n,k,5]$-LRC  with locality $r$, where  $k=n-\frac{n}{r+1}-3$ and  $n\ge \frac{{q\choose r+1}}{q^{r-1}-1}$. Hence, $N_q(5,r)\ge\Omega_r(q^2)$.
\item[{\rm (iv)}] {\rm (see Corollary \ref{cor:3.6})} Let $q$ be a power of $2$. For $4\le r\le q-1$ with $(r+1)|n$,  there exists an explicit construction of an optimal $q$-ary $[n,k,6]$-LRC  with locality $r$, where  $k=n-\frac{n}{r+1}-4$ and $n\ge \frac{{q\choose r+1}}{q^{r-1}-1}$. Hence, $N_q(6,r)\ge\Omega_r(q^2)$.
\end{itemize}
\end{theorem}
Lower bounds on $N_q(d,r)$ given in Theorem \ref{thm:1.2}(i) and (ii) are better than those in Lemma \ref{lem:1.1}, but Lemma \ref{lem:1.1} has less constraint on locality $r$. Furthermore, the bounds in Theorem \ref{thm:1.2} are constructive, while the bounds in Lemma \ref{lem:1.1} are algorithmically constructive.
\subsection{Organization}
In Section 2, we introduce some basic definitions and results on LRCs, binary constant weight codes and Moore matrices. In Section 3, we present our explicit construction of LRCs via binary constant weight codes. Furthermore, we apply this construction to various binary constant weight codes to obtain our main result.
\section{Preliminaries}
In this section, we introduce some basic definitions and results on LRCs, binary constant weight codes and Moore matrices.
\subsection{Locally repairable codes}\label{subsec:2.1}
Let $\F_q$ be the finite field of $q$ elements.
We denote by $[n]$ the set $\{1,2,\dots,n\}$. For a vector $\bu=(u_1,u_2,\dots,u_n)\in\F_q^n$ and a subset $I\subseteq[n]$, we denote by $\bu_I$ the projection of $\bu$ on $I$, i.e., $\bu_I=(u_i)_{i\in I}$. For a subset $C\subseteq\F_q^n$, we denote by $R_I$ the set $\{\bc_I:\; \bc\in C\}$.
\begin{defn}\label{def:1}{\rm
A $q$-ary block code $C\subseteq \F_q^n$ of length $n$ is called a locally recoverable code or locally repairable code (LRC for short) with locality $r$ if for any $i\in[n]$, there exists a subset $I\subset[n]\setminus\{i\}$ of size $r$ such that  for any $\bc=(c_1,\dots,c_n)\in C$, $c_i$ can be recovered by $\{c_j\}_{j\in I}$, i.e., for  any $i\in[n]$, there exists a subset $I\subset[n]\setminus\{i\}$ of size $r$ such that  for any  $\bu,\bv\in C$, $\bu_{I\cup\{i\}}=\bv_{I\cup\{i\}}$ if and only $\bu_I=\bv_I$. The set $I$ is called a recover set of $i$.
}\end{defn}
\begin{rmk}\label{rmk:1}{\rm
In literature, there are various definitions of LRCs and they are equivalent.  For instance, we have the following equivalence definitions of a locally recoverable code.
\begin{itemize}
\item[(i)] For any $i\in[n]$, there exists a subset $I\subset[n]\setminus\{i\}$ of size $r$ such that position $i$ of every codeword $\bc\in C$ is determined by $\bc_I$.
\item[(ii)] For any $i\in[n]$, there exists a subset $I\subset[n]\setminus\{i\}$ of size $r$ such that
\[C(i,\Ga)_I\cap C(i,\Gb)_I=\emptyset\]
for any $\Ga\neq\Gb\in\F_q$, where $C(i,\Ga)=\{\bc\in C:\; c_i=\Ga\}$.
\end{itemize}
}\end{rmk}

\begin{rmk}\label{rmk:2}{\rm If $C$ is a $q$-ary linear code with a parity-check matrix of the following form
\begin{equation}\label{eq:} H=\left(\begin{array}{c|c|c|c}
\bi&\bo&\cdots&\bo\\ \hline
\bo&\bi&\cdots&\bo \\ \hline
\vdots&\vdots&\ddots&\vdots \\ \hline
\bo&\bo&\cdots&\bi \\ \hline
D_1&D_2&\cdots&D_m
\end{array}
\right)
\end{equation}
for some matrices $D_i\in\F_q^{(n-k-n/(r+1))\times(r+1)}$,  where $\bi$ and $\bo$ stand for the all-one vector and the zero vector of length $r+1$, respectively, then $C$ is an $[n,\ge k]$-LRC with locality $r$. To show the locality $r$, let $\bc=(c_1,c_2,\dots,c_n)$ be a codeword and suppose we want to repair $c_i$. Write $i=(r+1)a+b$ for some $0\le a\le \frac{n}{r+1}-1$ and $1\le b\le r+1$. As $H$ is a parity-check matrix, we must have $\sum_{j=1}^{r+1}c_{(r+1)a+j}=0$, i.e., $c_i=-\sum_{1\le j\le r+1,j\neq b}c_{(r+1)a+j}$. This implies that $c_i$ is determined by $\bc_I$, where $I=\{(r+1)a+1, (r+1)a+2,\dots, (r+1)a+r+1\}\setminus\{i\}$.
}\end{rmk}

\subsection{Constant-weight codes}
A binary constant-weight code of length $n$ is a subset of $\F_2^n$ with each codeword having a fixed Hamming weight.
A binary constant-weight code of length $n$, size $M$, minimum distance $d$, and weight $w$ is denoted as $(n,M,d;w)$. It is a well-known fact that the followings are equivalent
\begin{itemize}
\item[(i)] There is a binary constant-weight code of length $n$, size $M$, weight $w$ and minimum distance at least $2w-2t$;
\item[(ii)] There is a set $\{I_i\}_{i=1}^M$ of subsets of $[n]$ such that $|I_i|=w$ and $|I_i\cap I_j|\le t$ for all $1\le i\neq j\le M$.
\end{itemize}

For  given $n,d,w$, it is a central coding problem to determine the maximum $M$ such that there is a binary $(n,M,d; w)$  constant weight code. In view of this,  we define
\[A(n,d,w):=\max \{M: \mbox{there exists a binary $(n,M,d; w)$  constant weight code}\}.\]

It is a challenging task to determine the exact value of $A(n,d,w)$ in general. Until now, the exact values of $A(n,d,w)$ have been determined for either some special values $n, d,w$ or some small $n, d,w$. Instead, researchers have made great effort on establishing some reasonable upper and lower bounds on $A(n,d,w)$ \cite{BSS}.

For our application, we are interested in constant weight codes of weight $r+1$ and minimum distance $2r$, namely the value $A(n,2r,r+1)$ only. Binary constant weight codes  are closely related to  Steiner systems (the reader may refer to \cite[Chapter 8]{AK92} for details on Steiner systems). Precisely speaking, we have the following result.
\begin{lemma}\label{lem:2.1} (see \cite{S66} and \cite[pp.528]{MS79}) There is a Steiner system $S(w-\Gd+1,w,n)$ if and only if one has
\begin{equation}
A(n,2\Gd,w)=\frac{{n\choose w-\Gd+1}}{{w\choose w-\Gd+1}}=\frac{n(n-1)\cdots(n-w+\Gd)}{w(w-1)\cdots\Gd}.
\end{equation}
In particular, a Steiner system $S(2,w,n)$ exists if and only if one has
\begin{equation}
A(n,2w-2,w)=\frac{n(n-1)}{w(w-1)}.
\end{equation}
\end{lemma}

Various values of $A(n,d,w)$ have been obtained via existence of Steiner systems. By using projective geometry, we obtain the following explicit construction of binary constant weight codes.
\begin{lemma}\label{lem:2.2} One has
\[A(\ell^t, 2\ell-2,\ell)=\frac{\ell^{t-1}(\ell^t-1)}{\ell-1}\]
for any prime power $\ell$ and integer $t\ge 1$. Furthermore, the above binary constant-weight codes can be explicitly constructed through Steiner system based on the projective geometry.
\end{lemma}

\begin{lemma}\label{lem:2.3}\cite{XING}
If $\delta\ge 3$ and $q$ is a prime power, then
\[ A(q,2\delta, w)\ge \frac{{q\choose w}}{q^{\delta-1}-1}.\]
In particular, for a prime power $q$ and an integer $r\ge 3$, one has
\begin{equation}
A(q,2r, r+1)\ge \frac{{q\choose r+1}}{q^{r-1}-1}=\Omega_r(q^2).
\end{equation}
\end{lemma}

\subsection{Moore determinant}
  \label{subsec:2.3}
  Let $\ell$ be a power of $q$. For elements $\Ga_1,\dots,\Ga_h\in\F_\ell$, the Moore matrix is defined by
  \[M=\left(\begin{array}{cccc}
\Ga_1&\Ga_2&\cdots&\Ga_h\\
\Ga_1^q&\Ga_2^q&\cdots&\Ga_h^q\\
\vdots&\vdots&\ddots&\vdots \\
\Ga_1^{q^{h-1}} &\Ga_2^{q^{h-1}} &\cdots&\Ga_h^{q^{h-1}}
\end{array}
\right)\in\F_{\ell}^{h\times h}.\]
  The determinant $\det(M)$ is given by the following formula
  \[\det(M)=\prod_{(c_1,\dots,c_h)}(c_1\Ga_1+\cdots+c_h\Ga_h),\]
  where $(c_1,\dots,c_h)$ runs through all non-zero direction vectors in $\F_q^h$. Thus, $\det(M)\neq 0$ if and  only if  $\Ga_1,\dots,\Ga_h$ are $\F_q$-linearly independent.

\section{Explicit construction}
In this section, we make use of the parity-check matrix of the form given in Remark \ref{rmk:2} to construct  LRCs of distances $5$ and $6$. Firstly, we present an important Lemma which is essential for the construction of  LRCs.

Let $I$ and $J$ be two subsets of $\F_q$ with size $r+1$ such that $|I\cap J|\le 1$. Denote $I=\{a_1,\cdots,a_{r+1}\}$ and $J=\{b_1,\cdots,b_{r+1}\}$.
Define the following two matrices over $\F_q$
{ \begin{equation}
A=\left(\begin{array}{cccc}
1&1&\cdots&1\\
0&0&\cdots&0\\
a_1&a_2&\cdots&a_{r+1}\\
a_1^2&a_2^2&\cdots&a_{r+1}^2\\
a_1^3&a_2^3&\cdots&a_{r+1}^3\\
\end{array}
\right)
\quad \quad
B=\left(\begin{array}{cccc}
0&0&\cdots&0\\
1&1&\cdots&1\\
b_1&b_2&\cdots&b_{r+1}\\
b_1^2&b_2^2&\cdots&b_{r+1}^2\\
b_1^3&b_2^3&\cdots&b_{r+1}^3\\
\end{array}
\right).
\end{equation}}

Write $A=[\a_1,\cdots, \a_{r+1}]$ and  $B=[\b_1,\cdots, \b_{r+1}]$, where $\a_i, \b_i$ are column vectors.
\begin{lemma}\label{lem:3.0}  Let $A$ and $B$ be the two matrices defined above.
Then any four column vectors consisting of two columns from  matrix A and the other two columns from matrix B are linearly independent.
\end{lemma}
\begin{proof}
Without loss of generality, we prove that $\a_1, \a_2,\b_1, \b_2$ are linearly independent.
Suppose that $\Gl_i\in\F_q$ such that $\Gl_1\a_1+\Gl_2 \a_2+\Gl_3\b_1+\Gl_4\b_2=\bo$. It is clear that $\Gl_1+\Gl_2=\Gl_3+\Gl_4=0$. If one of $\Gl_i$ is zero, say $\Gl_1=0$, then $\Gl_2=0$ as well. This gives $\Gl_3\b_1+\Gl_4\b_2=\bo$. This implies that $\Gl_3=\Gl_4=0$ as $\b_1$ and $\b_2$ are linearly independent. Suppose that none of $\Gl_i$ were zero. Put $a=\Gl_1$ and $b=\Gl_3$, then $\Gl_2=-a$ and $\Gl_4=-b$. Thus, considering the last three coordinates of  the column vector $\Gl_1\a_1+\Gl_2 \a_2+\Gl_3\b_1+\Gl_4\b_2$ gives the following three identities

\begin{eqnarray*}
a(a_1-a_2)&=&b(b_1-b_2)\\
 a(a_1^2-a_2^2)&=&b(b_1^2-b_2^2)\\
  a(a_1^3-a_2^3)&=&b(b_1^3-b_2^3).
\end{eqnarray*}
 Dividing the second and third identities by the first identity in the above display gives
 \begin{equation}\label{eq:8}
 a_1+a_2=b_1+b_2,\quad a_1^2+a_1a_2+a_2^2=b_1^2+b_1b_2+b_2^2.
 \end{equation}
The equations \eqref{eq:8} are equivalent to
 \begin{equation}
 a_1+a_2=b_1+b_2,\quad  a_1a_2=b_1b_2.
 \end{equation}
This implies that both  $\{a_1,a_2\}$ and $\{b_1,b_2\}$ are  the two roots of the same quadratic equation, i.e., $\{a_1,a_2\} =\{b_1,b_2\}$ . On the other hand, $\{a_1,a_2\}\subseteq I$ and $\{b_1,b_2\}\subseteq J$. This implies that $|I\cap J|\ge 2$ which is a  contradiction.
\end{proof}

\begin{theorem} \label{thm:3.1} Let $r\ge 4$ be an integer.
If there is a binary $(q,m,2r;r+1)$ constant weight code $A$, then there exists an optimal $q$-ary $[n,k,5]$-LRC $C$ with locality $r$, where $n=(r+1)m$ and $k=n-m-3$. Furthermore, $C$ can be explicitly constructed as long as $A$ is explicitly given.
\end{theorem}
\begin{proof}
As there is  a binary $(q,m,2r;r+1)$ constant weight code $A$, one has a family $\{I_i\}_{i=1}^m$ of subsets of $\F_q$ such that $|I_i|=r+1$ and $|I_j\cap I_j|\le 1$ for all $1\le i\neq j\le m$. Label  elements of $I_i$ by $\{\Ga_{i1},\Ga_{i2},\dots,\Ga_{i,r+1}\}$.
Define the following $3\times (r+1)$ matrices
\begin{equation}
D_i=\left(\begin{array}{cccc}
\Ga_{i1}&\Ga_{i2}&\cdots&\Ga_{i,r+1}\\
\Ga_{i1}^2&\Ga_{i2}^2&\cdots&\Ga_{i,r+1}^2\\
\Ga_{i1}^3&\Ga_{i2}^3&\cdots&\Ga_{i,r+1}^3
\end{array}
\right)
\end{equation}
for $i=1,2,\dots,m$.
Now we define a $(3+m)\times n$ matrix
\begin{equation}\label{eq:} H=\left(\begin{array}{c|c|c|c}
\bi&\bo&\cdots&\bo\\ \hline
\bo&\bi&\cdots&\bo \\ \hline
\vdots&\vdots&\ddots&\vdots \\ \hline
\bo&\bo&\cdots&\bi \\ \hline
D_1&D_2&\cdots&D_m
\end{array}
\right),
\end{equation}
where $\bi$ and $\bo$ stand for the all-one vector and the zero vector of length $r+1$, respectively.

We claim that the $q$-ary linear code $C$ with $H$ as a parity-check matrix is the desired optimal $q$-ary $[n,k,5]$-LRC with locality $r$. Length and locality are clear. The dimension of $C$ is at least $n-m-3=k$. Thus, we may assume that  the dimension of $C$ is  $k$ (otherwise one can increase rows of $H$ if the dimension of $C$ is less than $k$).  By the Singleton bound, the minimum distance is upper bounded by
\[d\le  n-k-\left\lceil\frac kr\right\rceil+2=5.\]
Thus, it is sufficient to show that the minimum distance is at least $5$, i.e., any four columns of $H$ are $\F_q$-linearly independent.  We are going to prove $d=5$ in five cases.

Note that every column of $H$ can be indexed by a pair $(i,j)$ for $1\le i\le m$ and $1\le j\le r+1$ with the $(i,j)$th column $\bh_{ij}=(0,\cdots,0,1,0,\cdots,0,\Ga_{ij},\Ga_{ij}^2,\Ga_{ij}^3)^T$, where $1$ is located at position $i$. We say that $\bh_{ij}$ and $\bh_{ts}$ belong to the same block if $i=t$.

Consider the four columns $\bh_{i_1,j_1}$, $\bh_{i_2,j_2}$, $\bh_{i_3,j_3}$ and $\bh_{i_4,j_4}$ and define the matrix $D=( \bh_{i_1,j_1}, \bh_{i_2,j_2}, $ $\bh_{i_3,j_3},\bh_{i_4,j_4})$.

{\it Case (i)} $\{\bh_{i_t,j_t}\}_{t=1}^4$ belong to the same block, i.e, $i_1=i_2=i_3=i_4$. Then it is clear that they are $\F_q$-linearly independent as  the $i_1$th row together with the last three rows of $D$ forms a Vandermanond matrix. Hence, they are $\F_q$-linearly independent.

{\it Case (ii)}  $\{\bh_{i_t,j_t}\}_{t=1}^4$ belong to four distinct blocks. Then they are $\F_q$-linearly independent as  rows $i_1$ to $i_4$ of $D$ from the $4\times 4$ identity matrix.

{\it Case (iii)}  $\{\bh_{i_t,j_t}\}_{t=1}^4$ belong to three distinct blocks. Without loss of generality, we may assume that $1\le i_1=i_2<i_3<i_4\le m$. Suppose that $\Gl_t\in\F_q$ such that $\sum_{t=1}^4\Gl_t\bh_{i_t,j_t}=\bo$. Then positions $i_3$ and $i_4$ of the column vector $\sum_{t=1}^4\Gl_t\bh_{i_t,j_t}$ are equal to $\Gl_3$ and $\Gl_4$, respectively. Hence, $\Gl_3=\Gl_4=0$. Thus, we have $\Gl_1\bh_{i_1,j_1}+\Gl_2\bh_{i_2,j_2}=\bo$. This implies that $\Gl_1=\Gl_2=0$ as $\bh_{i_1,j_1}$ and $\bh_{i_2,j_2}$ belong to the same block and hence are linearly independent by Case (i).

{\it Case (iv)}  Three of $\{\bh_{i_t,j_t}\}_{t=1}^4$ belong to the same block and the other one lies in a different block. Without loss of generality, we may assume that $1\le i_1=i_2=i_3<i_4\le m$. Suppose that $\Gl_t\in\F_q$ such that $\sum_{t=1}^4\Gl_t\bh_{i_t,j_t}=\bo$. Then position  $i_4$ of  the column vector $\sum_{t=1}^4\Gl_t\bh_{i_t,j_t}$ is equal to  $\Gl_4$. Hence, $\Gl_4=0$. Thus, we have $\sum_{t=1}^3\Gl_t\bh_{i_t,j_t}=\bo$. This implies that $\Gl_1=\Gl_2=\Gl_3=0$ as $\bh_{i_1,j_1}$, $\bh_{i_2,j_2}$ and $\bh_{i_3,j_3}$ belong to the same block are hence are linearly independent by Case (i).

{\it Case (v)}  Two of $\{\bh_{i_t,j_t}\}_{t=1}^4$ belong to one block and the other two  lie in a different block. Then they are also linearly independent by Lemma \ref{lem:3.0}.\end{proof}

Theorem \ref{thm:3.1} shows that one can construct an LRC  with distance 5  as long as there exists a constant-weight code with the required parameters. Now we are going to give another construction of LRC  with distance 6. Similarly, we present the following lemma first which is crucial for the construction.

Again let $I$ and $J$ be two subsets of $\F_q$ with size $r+1$ such that $|I\cap J|\le 1$. Denote $I=\{a_1,\cdots,a_{r+1}\}$ and $J=\{b_1,\cdots,b_{r+1}\}$.
Define the following two matrices over $\F_q$.
{\begin{equation}
A'=\left(\begin{array}{cccc}
1&1&\cdots&1\\
0&0&\cdots&0\\
a_1&a_2&\cdots&a_{r+1}\\
a_1^2&a_2^2&\cdots&a_{r+1}^2\\
a_1^3&a_2^3&\cdots&a_{r+1}^3\\
a_1^4&a_2^4&\cdots&a_{r+1}^4\\
\end{array}
\right)
\quad \quad
B'=\left(\begin{array}{cccc}
0&0&\cdots&0\\
1&1&\cdots&1\\
b_1&b_2&\cdots&b_{r+1}\\
b_1^2&b_2^2&\cdots&b_{r+1}^2\\
b_1^3&b_2^3&\cdots&b_{r+1}^3\\
b_1^4&b_2^4&\cdots&b_{r+1}^4\\
\end{array}
\right).
\end{equation}}
Denote by $A'=[\a_1',\cdots, \a_{r+1}']$ and  $B'=[\b_1',\cdots, \b_{r+1}']$ where $\a_i', \b_i'$ are column vectors.
\begin{lemma}\label{lem:3.1} Let $A',B'$ be the two matrices defined above. If $q$ is a power of $2$, then
any five column vectors consisting of three columns from  $A'$ and the other two columns from $B'$ are linearly independent.
\end{lemma}

\begin{proof}
Without loss of generality, we only need to prove $\a_1', \a_2', \a_3', \b_1', \b_2'$ are linearly independent.

Suppose that $\Gl_t\in\F_q$ such that $\Gl_1\a_1'+\Gl_2 \a_2'+\Gl_3\a_3'+\Gl_4\b_1'+\Gl_5\b_2'=\bo$.
If one of $\Gl_t$ is $0$, then proof  is reduced to that of  case (iv) or (v) of Theorem \ref{thm:3.1}. Now suppose that none of $\Gl_t$ were $0$.
By considering the first and second coordinates of $\Gl_1\a_1'+\Gl_2 \a_2'+\Gl_3\a_3'+\Gl_4\b_1'+\Gl_5\b_2'$, we have $\Gl_1+\Gl_2+\Gl_3=\Gl_4+\Gl_5=0$.
Let $\Gl_2=a$, $\Gl_3=b$ and $\Gl_4=c$. Then $\Gl_1=a+b$ and $\Gl_5=c$.

Thus, the identity $\Gl_1\a_1'+\Gl_2 \a_2'+\Gl_3\a_3'+\Gl_4\b_1'+\Gl_5\b_2'=\bo$ becomes
$a(\a_1'+\a_2')+b(\a_1'+\a_3')=c(\b_1'+\b_2')$. Therefore, $\a_1'+\a_2'$, $\a_1'+\a_3'$ and $\b_1'+\b_2'$ are $\F_q$-linearly dependent. This implies that the following matrix
 \[B=\left(\begin{array}{ccc}
a_1+a_2& a_1+a_3& b_1+b_2\\
a_1^2+a_2^2& a_1^2+a_3^2& b_1^2+b_2^2\\
a_1^3+a_2^3& a_1^3+a_3^3& b_1^3+b_2^3\\
a_1^4+a_2^4& a_1^4+a_3^4& b_1^4+b_2^4\\
 \end{array}
 \right)\]
 consisting of the last four positions of these three vectors $\a_1'+\a_2'$, $\a_1'+\a_3'$ and $\b_1'+\b_2'$ has rank at most $2$. Thus, the $3\times 3$ submatrix of $B$ consisting of rows $1$, $2$ and $4$   has rank at most $2$ as well. Since this submatrix is a Moore matrix, $a_1+a_2$, $a_1+a_3$ and $b_1+b_2$ are $\F_2$-linearly dependent. This implies that $b_1+b_2$  is equal to $a_1+a_2$, $a_1+a_3$ or $a_2+a_3$. Without loss of generality, we may assume that  $b_1+b_2$  is equal to $a_1+a_2$. Subtracting the third column by the first column of $B$, one gets a matrix
 \[B_1=\left(\begin{array}{ccc}
 a_1+a_2& a_1+a_3&0\\
a_1^2+a_2^2& a_1^2+a_3^2& 0\\
a_1^3+a_2^3& a_1^3+a_3^3& (a_1+a_2)(a_1a_2+b_1b_2)\\
a_1^4+a_2^4& a_1^4+a_3^4& 0\\
 \end{array}
 \right).\]
 Note that the elements at entry $(3,3) $ of $B_1$ is
 \[b_1^3+b_2^3-(a_1^3+a_2^3)=b_1^3+b_2^3+a_1^3+a_2^3=(a_1+a_2)^3+(b_1+b_2)^3+(a_1+a_2)(a_1a_2+b_1b_2)=(a_1+a_2)(a_1a_2+b_1b_2).\]
As the submatrix of $B_1$ consisting of rows $1$ and $2$ and columns $1$ and $2$ is a $2\times 2$ Moore matrix and $a_1+a_2$, $a_1+a_3$  are $\F_2$-linearly independent, the first two rows of $B_1$ are $\F_q$-linearly independent. This forces that $(a_1+a_2)(a_1a_2+b_1b_2)=0$, i.e., $a_1a_2=b_1b_2$. Combining with the fact that $b_1+b_2= a_1+a_2$, we must have $\{a_1,a_2\}=\{b_1,b_2\}$. This is a contradiction since $|I\cap J|\le 1$.
\end{proof}

\begin{theorem} \label{thm:3.2} Let $r\ge 5$ be an integer and let $q$ be a power of $2$.
If there is a binary $(q,m,2r; r+1)$ constant weight code $A$, then there exists an optimal $q$-ary $[n,k,6]$-LRC $C$ with locality $r$, where $n=(r+1)m$ and $k=n-m-4$. Furthermore, $C$ can be explicitly constructed as long as $A$ is explicitly given.
\end{theorem}
\begin{proof} Again, let $\{I_i\}_{i=1}^m$  be a family of subsets of $\F_q$ such that $|I_i|=r+1$ and $|I_j\cap I_j|\le 1$ for all $1\le i\neq j\le m$. Label  elements of $I_i$ by $\{\Ga_{i1},\Ga_{i2},\dots,\Ga_{i,r+1}\}$.
Define the following $4\times (r+1)$ matrices
\begin{equation}
D_i=\left(\begin{array}{cccc}
\Ga_{i1}&\Ga_{i2}&\cdots&\Ga_{i,r+1}\\
\Ga_{i1}^2&\Ga_{i2}^2&\cdots&\Ga_{i,r+1}^2\\
\Ga_{i1}^3&\Ga_{i2}^3&\cdots&\Ga_{i,r+1}^3\\
\Ga_{i1}^4&\Ga_{i2}^4&\cdots&\Ga_{i,r+1}^4
\end{array}
\right)
\end{equation}
for $i=1,2,\dots,m$.
Now we define a $(4+m)\times n$ matrix
\begin{equation}\label{eq:} H=\left(\begin{array}{c|c|c|c}
\bi&\bo&\cdots&\bo\\ \hline
\bo&\bi&\cdots&\bo \\ \hline
\vdots&\vdots&\ddots&\vdots \\ \hline
\bo&\bo&\cdots&\bi \\ \hline
D_1&D_2&\cdots&D_m
\end{array}
\right),
\end{equation}
where $\bi$ and $\bo$ stand for the all-one vector and the zero vector of length $r+1$, respectively.

We claim that the $q$-ary linear code $C$ with $H$ as a parity-check matrix is the desired optimal $q$-ary $[n,k,6]$-LRC with locality $r$. Length and locality are clear. The dimension of $C$ is at least $n-m-4=k$. Thus, we may assume that  the dimension of $C$ is  $k$ (otherwise one can increase rows of $H$ if the dimension of $C$ is less than $k$).  By the Singleton bound, the minimum distance is upper bounded by
\[d\le  n-k-\left\lceil\frac kr\right\rceil+2=6.\]
Thus, it is sufficient to show that the minimum distance is at least $6$, i.e., any five columns of $H$ are $\F_q$-linearly independent.

Now every column of $H$ is indexed by a pair $(i,j)$ with $1\le i\le m$ and $1\le j\le r+1$ with the $(i,j)$th column $\bh_{ij}=(0,\cdots,0,1,0,\cdots,0,\Ga_{ij},\Ga_{ij}^2,\Ga_{ij}^3,\Ga_{ij}^4)^T$, where $1$ is located at position $i$.

As in Theorem \ref{thm:3.1}, to show that any five columns of $H$ are $\F_q$-linearly independent, we can discuss different cases. In this theorem, we have the following cases: (i) all five columns belong to the same block; (ii) four columns belong to one block and the remaining column lies in a different block; (iii) five columns lie in three different blocks; (iv) five columns lie in four different blocks; (v) five columns lie in five different blocks; (vi) three columns belong to one block and the remaining two columns lie in anthor block.

For cases (i)-(v), one can prove it by using the similar arguments as in Theorem \ref{thm:3.1}. Case (vi) follows from the results in Lemma \ref{lem:3.1}.

\end{proof}

From Theorems \ref{thm:3.1} and \ref{thm:3.2}, we can see that  one can construct an LRC  with distance 5 or 6 as long as there exists a constant-weight code with the required parameters. The topic of constant-weight codes has been studied a lot and many results have been known. Therefore, we can make use of the known results on the construction of constant-weight codes to produce LRCs via Theorems \ref{thm:3.1} and \ref{thm:3.2}.

Combining Theorem \ref{thm:3.1} with Lemma \ref{lem:2.2} gives the following result immediately.
\begin{cor}\label{cor:3.3} If $r+1\ge 5$ is a prime power, then for any $t\ge 1$ there exists an explicit construction of a family of optimal $q$-ary $[n,k,5]$-LRCs  with locality $r$, where $q=(r+1)^t$, $n=(r+1)m$ and $k=n-m-3$ and $m=\frac{(r+1)^{t-1}((r+1)^t-1)}{r}$. Hence, $n=\frac 1r q(q-1)$.
\end{cor}

Combining Theorem \ref{thm:3.1} with Lemma \ref{lem:2.2} gives the following result immediately.
\begin{cor}\label{cor:3.4} If $r+1\ge 8$ is a power of $2$, then for any $t\ge 1$ there exists an explicit construction of a family of optimal $q$-ary $[n,k,6]$-LRCs  with locality $r$, where $q=(r+1)^t$, $n=(r+1)m$ and $k=n-m-4$ and $m=\frac{(r+1)^{t-1}((r+1)^t-1)}{r}$. Hence, $n=\frac 1r q(q-1)$.
\end{cor}

Combining Theorem \ref{thm:3.2}  with Lemma \ref{lem:2.3} gives the following result immediately.
\begin{cor}\label{cor:3.5} For $4\le r\le q-1$ with $(r+1)|n$,  there exists an explicit construction of an optimal $q$-ary $[n,k,5]$-LRC  with locality $r$, where  $k=n-\frac{n}{r+1}-3$ and $n=\Omega_r(q^2)$.
\end{cor}

Combining Theorem \ref{thm:3.2}  with Lemma \ref{lem:2.3} gives the following result immediately.
\begin{cor}\label{cor:3.6} Let $q$ be a power of $2$. For $4\le r\le q-1$ with $(r+1)|n$,  there exists an explicit construction of an optimal $q$-ary $[n,k,6]$-LRC  with locality $r$, where  $k=n-\frac{n}{r+1}-4$ and $n=\Omega_r(q^2)$.
\end{cor}

By applying Theorems \ref{thm:3.1} and \ref{thm:3.2} to various other binary constant weight, one could obtain more optimal LRCs with distance $d=5$ or $6$.


\begin{thebibliography}{12}
\bibitem{AK92} E. F. Assmus and J. D. Key, {\it Designs and Their Codes,}  Cambridge University Press, 1992.

\bibitem{Ball}
S. Ball, {\it On large subsets of a finite vector space in which every subset of basis size is a basis}, J. Eur. Math. Soc., vol. 14, pp. 733-748, Oct. 2012.



\bibitem{BTV17} A. Barg, I. Tamo, and S. Vl\u{a}du\c{t}, {\it Locally recoberable codes on algebraic curves}, IEEE Trans. Inform. Theory {\bf 63}(8)(2017), 4928--4939.




\bibitem{BSS} A. Brouwer, J. Shearer, N. Sloane and W. Smith, {\it A new table of constant weight code,} IEEE Transactions on Information Theory, {\bf 36}(1990), 1334-1380.






\bibitem{FY14} M. Forbes and S. Yekhanin, {\it On the locality of codeword symbols in non-linear codes,} Discrete Mathematics
{\bf 324(6)}(2014), 78--84.


\bibitem{GHSY12} P. Gopalan, C. Huang, H. Simitci and S. Yekhanin, {\it On the locality of codeword symbols}, IEEE Trans. Inf. Theory {\bf 58}(11)(2012), 6925--6934.

\bibitem{GXY18} V. Guruswami, C. Xing and C. Yuan, {\it How long can optimal locally repairable codes be?} to appear in Proceedings of RANDOM 2018 (see arXiv: https://arxiv.org/abs/1807.01064).


\bibitem{HL07} J. Han and L. A. Lastras-Montano, {\it Reliable memories with subline accesses}, Proc. IEEE Internat. Sympos. Inform. Theory, 2007,  2531--2535.



\bibitem{HCL} C. Huang, M. Chen, and J. Li, {\it Pyramid codes: Flexible schemes to trade space for access efficiency in reliable data storage systems,} Sixth IEEE International Symposium on Network Computing and Applications, 2007, 79--86.



\bibitem{JMX17} L. Jin, L. Ma and C, Xing, {\it Construction of optimal locally repairable codes via automorphism groups of rational function fields}, https://arxiv.org/abs/1710.09638.

\bibitem{KBTY18} O. Kolosov, A. Barg, I. Tamo and G. Yadgar, {\it Optimal LRC codes for all lengths $n\leq q$}, https://arxiv.org/pdf/1802.00157.


\bibitem{MS79} F. J. MacWilliams and N. J. A. Sloane,  ``The Theory uf Error-Correcting Codes,"  Amsterdam: North-Holland, 1979.


\bibitem{LMX17}
X. Li, L. Ma and C. Xing, {\it Optimal locally repairable codes via elliptic curves,} To appear in IEEE Trans. Inf. Theory (see https://arxiv.org/abs/1712.03744).

\bibitem{LXY} Y. Luo, C. Xing and C. Yuan, {\it Optimal locally repairable codes of distance $3$ and $4$ via cyclic codes}, To appear in IEEE Trans. Inf. Theory (see  https://arxiv.org/abs/1801.03623).

\bibitem{PD14} D. S. Papailiopoulos and A.G. Dimakis, {\it Locally repairable codes}, IEEE Trans. Inf. Theory {\bf 60}(10)(2014), 5843--5855.

\bibitem{PKLK12} N. Prakash, G.M. Kamath, V. Lalitha and P.V. Kumar, {\it Optimal linear codes with a local-error-correction property}, Proc. 2012 IEEE Int. Symp. Inform. Theory, 2012, 2776--2780.



\bibitem{S66} J. Sch\"{o}nheim. {\it On maximal systems of $k$-tuples,} Studia Scient. Malh. Hungar., Vol. 1(1966), pp. 363-368.

 \bibitem{SRKV13} N. Silberstein, A.S. Rawat, O.O. Koyluoglu and S. Vichwanath, {\it Optimal locally repairable codes via rank-matric codes}, Proc. IEEE Int. Symp. Inf. Theory,  2013, 1819--1823.

\bibitem{TB14} I. Tamo and A. Barg, {\it A family of optimal locally recoverable codes}, IEEE Trans. Inform. Theory {\bf 60}(8)(2014), 4661--4676.

\bibitem{TPD16} I. Tamo, D.S. Papailiopoulos and A.G. Dimakis, {\it Optimal locally repairable codes and connections to matroid theory},  IEEE Trans. Inform. Theory {\bf 62}(12)(2016), 6661--6671.





\bibitem{XING} C. Xing and J. Ling, {\it A construction of binary constant-weight codes form algebraic curves over finite fields,} IEEE Transactions on Information Theory, {\bf 51}(2005), 3674--3678.
\end{thebibliography}
\end{document}